\documentclass[a4paper,showkeys]{revtex4}
\usepackage[utf8x]{inputenc}
\usepackage{amsmath}
\usepackage{amsfonts}
\usepackage{amssymb}
\usepackage{amsthm}
\usepackage{braket}
\usepackage{graphicx}
\usepackage{bbm}
\usepackage{marvosym}
\usepackage{mathtools}
\usepackage{chngcntr}
\usepackage[noautoscale]{youngtab}
\newcommand{\defeq}{\vcentcolon=}

\newtheorem{defi}{Definition}

\newtheorem{thm}{Theorem}
\newtheorem{theo}{Theorem}[section]
\newtheorem{cor}{Corollary}

\begin{document}

\title{How not to R\'enyi-generalize the Quantum Conditional Mutual Information}

\author{Paul Erker}
\email{paul.erker@gmail.com}

\affiliation{Institute of Theoretical Physics, ETH Z\"urich, Wolfgang-Pauli-Str.~27, 8093 Z\"urich, Switzerland.}
\affiliation{F\'{\i}sica Te\`{o}rica: Informaci\'{o} i Fen\`{o}mens  Qu\`{a}ntics, Universitat Aut\`{o}noma de Barcelona, ES-08193 Bellaterra (Barcelona), Spain}
\affiliation{Faculty of Informatics, Universit\`{a} della Svizzera italiana, Via G. Buffi 13, 6900 Lugano, Switzerland}

\date{\today}

\begin{abstract} 
We study the relation between the quantum conditional mutual information and the quantum $\alpha$-R\'enyi divergences. Considering the totally antisymmetric state we show that it is not possible to attain a proper generalization of the quantum conditional mutual information by optimizing the distance in terms of quantum $\alpha$-R\'enyi divergences over the set of all Markov states. The failure of the approach considered arises from the observation that a small quantum conditional mutual information does not imply that the state is close to a quantum Markov state.
\end{abstract}

\maketitle

\section{Introduction}

Mutual information measures are widely used to characterize the correlations in quantum many-body systems as well as in classical systems. Posing the question on how much can be deduced from a system  about a second system conditioned on knowing the state of a third system has led to the definition of the conditional mutual information, i.e.
\begin{equation}
I(A:B|C)_{\rho}  \equiv H (AC)_{\rho} + H(BC)_{\rho} - H(C)_{\rho} - H(ABC)_{\rho}.
\end{equation}
 In the setting where one is provided with i.i.d.~(identically and independently distributed) resources this is a well studied quantity and it has been shown that for the case of quantum state redistribution it characterizes the optimal amount of qubits needed in order to accomplish this task \cite{DY07}. Furthermore it is the basis for the squashed entanglement, an important quantity in entanglement theory \cite{CW03}. 
The $\alpha$-R\'enyi divergences on the other hand have proven themselves useful in the task of generalizing quantities of information to a setting beyond i.i.d.~ \cite{Dat09,Ren05,Tom12}. Nevertheless, the question of how to R\'enyi-generalize the quantum conditional mutual information to this setting remains an open problem. \\

In the following we will elucidate on how not to do it. In \cite{LMW12} the authors show that an approach of R\'enyi-generalizing the quantum conditional mutual information by simply taking a linear combination of R\'enyi  entropies in the form of
\begin{equation}
I_{\alpha}'(A:B|C)_{\rho}  \equiv H_{\alpha}(AC)_{\rho} + H_{\alpha}(BC)_{\rho} - H_{\alpha}(C)_{\rho} - H_{\alpha}(ABC)_{\rho},
\end{equation}
leads to a quantity that in general is not non-negative and also does not come equipped with a data-processing inequality. Therefore this quantity is not suitable for a proper generalization. 
In contrast, an approach where one optimizes the quantum $\alpha$-R\'enyi divergences over the set of all Markov states, i.e. states for which the quantum conditional mutual information is zero, may seem reasonable. As for the case of $D_{min}$($\tilde{D}_{\frac{1}{2}}$ in \cite{WWY14,Mul13}) it leads to a non-negative quantity that fulfills a data-processing inequality and a duality relation, as one would expect from a generalization of the quantum conditional mutual information.   However it will be shown in the following that this approach fails. This failure arises from the fact that in the case of the totally antisymmetric state this approach leads to a constant lower bound for the whole range of the parameter $\alpha$ while the quantum conditional mutual information goes to zero for large dimensions (with $\mathcal{O}(\frac{1}{d})$). Thus this approach cannot produce a quantity that in general lower bounds the quantum conditional mutual information, as it would be expected from a proper generalization \cite{Tom12}. \\

The article is organized as follows. In section II the basic physical and mathematical notions will be introduced. Section III contains the main results. These will be discussed in section IV, where we will also state open questions regarding the problem of R\'enyi-generalizing the quantum conditional mutual information.

\section{Definitions}

We make use of quantum systems A,B with corresponding finite-dimensional Hilbert spaces $\mathcal{H}_{A},\mathcal{H}_{B} $ where $|A|,|B| $ denote their dimensionality. According to the postulates of quantum mechanics, the Hilbert space of a composite system is given by the tensor product $\mathcal{H}_{AB}\defeq \mathcal{H}_{A}\otimes\mathcal{H}_{B}$. The state of a quantum system is represented by a density operator $\rho$ acting on some Hilbert space $\mathcal{H}$. In the following $\mathcal{D} (\mathcal{H})\defeq \{ \rho : \rho \geq 0, 0 < \textnormal{Tr}(\rho) \leq 1 \}$  denotes the set of quantum states acting on $\mathcal{H}$ and  $\mathcal{D}_{=} (\mathcal{H})\defeq \{  \rho : \rho \geq 0, \textnormal{Tr}(\rho) = 1 \}$ the set of all normalized states. To avoid ambiguities we let $\rho_A \equiv \textnormal{Tr}_B \, \rho_{AB}$ be the reduced state on system A.\\
To quantify the distance between any two quantum states $\rho , \sigma \in \mathcal{D} (\mathcal{H})$ we use the generalized fidelity \cite{Tom10}, defined as 
\begin{equation}
F (\rho , \sigma)\defeq  \left|\left| \sqrt{\rho} \sqrt{\sigma} \right|\right|_{1} - \sqrt{(1-\textnormal{Tr}\rho )(1-\textnormal{Tr}\sigma)}.
\end{equation}
In the case that either $\rho$ or $\sigma $ is an element of the set of all normalized states  $\mathcal{D}_{=}$ the above expression reduces to the standard fidelity
\begin{equation}
F (\rho , \sigma)\defeq  \left|\left| \sqrt{\rho} \sqrt{\sigma} \right|\right|_{1}.
\end{equation}
\\
Many of the well-known information measures in the asymptotic framework of quantum information theory can be written in terms of the quantum relative entropy or Umegaki relative entropy \cite{umegaki};

\begin{defi}
Let $\rho,\sigma \in \mathcal{D} (\mathcal{H}) $ with $supp\,\rho \subseteq supp\,\sigma$, then the quantum relative entropy is defined as
\begin{align}
D_1 (\rho \|  \sigma) \defeq \frac{1}{\textnormal{Tr} \rho} \textnormal{Tr} \left[ \rho \left( \textnormal{log} \rho - \textnormal{log} \sigma \right) \right] .
\end{align}
\end{defi}
Recalling the definitions for the von Neumann entropy $H(A)_{\rho}$, the conditional von Neumann entropy $H(A|B)_{\rho}$ and the von Neumann Mutual Information $I(A:B)_{\rho}$, it is straightforward to check that
\begin{align}
H(A)_{\rho} &\equiv - \textnormal{Tr} \, \rho \,\textnormal{log} \,\rho \\
&= - D_1 (\rho_A \|  \mathbbm{1}_A ),
\end{align}
\begin{align}
H(A|B)_{\rho} &\equiv H(AB)_{\rho} - H(B)_{\rho} \\
&= - \min_{\sigma_B} D_1 (\rho_{AB} \|  \mathbbm{1}_A \otimes \sigma_B )
\end{align}
and
\begin{align}
I(A:B)_{\rho} &\equiv H(A)_{\rho} + H(B)_{\rho} - H(AB)_{\rho} \\
&= \min_{\sigma_B} D_1 (\rho_{AB} \|  \rho_A \otimes \sigma_B ).
\end{align}

Generalizations of these quantities in terms of quantum R\'enyi divergences have found applications in various operational tasks \cite{Tom12,KRS09,MH11};
\begin{defi}
Let $\alpha \in (0,1) \cup (1,\infty)$ and let $\rho, \sigma \in \mathcal{D} (\mathcal{H})$ with $supp\,\rho \subseteq supp\,\sigma$. Then the quantum Rényi divergence of order $\alpha$ is defined as
\begin{equation}
D_{\alpha} (\rho \|  \sigma) \defeq  \frac{1}{\alpha - 1}\textnormal{log} \frac{ \textnormal{Tr} \rho^{\alpha} \sigma^{1- \alpha}}{\textnormal{Tr} \rho}.
\end{equation}
\end{defi}

Note that $D_{\alpha} $ is a monotonically increasing function in $\alpha$ and that in the limit $\alpha \rightarrow 1$ we recover the quantum relative entropy.\\
Furthermore we define the 0-relative entropy ($D_{\min}$ in Ref. \cite{Dat09}), which naturally appears in binary hypothesis testing when the probability for type I error is set to zero.

\begin{defi}
Let $\rho,\sigma \in \mathcal{D} (\mathcal{H})$, then the 0-relative entropy is
\begin{equation}
D_0 (\rho \|  \sigma) \defeq \lim_{\alpha \rightarrow 0^{+}} D_{\alpha} (\rho \|  \sigma) = - \textnormal{log} \,\textnormal{Tr}\, P_{\rho} \, \sigma,
\end{equation}
where $P_{\rho}$ denotes the projector onto the support of $\rho$.
\end{defi}
As above, the quantum conditional mutual information for a tripartite quantum state ABC is given by
\begin{equation}
I(A:B|C)_{\rho}  \equiv H(AC)_{\rho} + H(BC)_{\rho} - H(C)_{\rho} - H(ABC)_{\rho} .
\end{equation}
Moreover we define the set of all Markov states as
\begin{equation}
\mathcal{M}_{AB|C} \defeq \{ \sigma \in \mathcal{D} (\mathcal{H}_{ABC}): I (A:B |C)_{\sigma}=0\}.
\end{equation}
The condition that $\sigma$ fulfills strong subadditivity with equality (i.e. $ I (A:B |C)_{\sigma}=0$) is equivalent to the statement that there exists a decomposition of system C as
\begin{equation}
\mathcal{H}_C = \bigoplus_i  \mathcal{H}_{C^{L}_i} \otimes \mathcal{H}_{C^{R}_i} 
\end{equation}
into a direct sum of tensor products such that
\begin{equation}
\label{markov}
 \sigma_{ABC} = \bigoplus_i p_i \sigma_{AC^{L}_i} \otimes \sigma_{C^{R}_i B}\,,
\end{equation}
 where $\{p_i\}$ is some probability distribution (see Ref. \cite{HJPW04}). Furthermore, we define the totally antisymmetric state as 
\begin{equation}
\label{anti}
 \gamma_d = \frac{2}{d(d-1)} P_{\,\yng(1,1)}
\end{equation}
where $P_{\,\yng(1,1)}$ denotes the projector onto the antisymmetric subspace $ \Lambda^2 (\mathbb{C}^d)$ in $(\mathbb{C}^d)^{\otimes 2}$  .

\section{Main Result}
In Ref. \cite{ILW06} the authors show that the quantity 
\begin{equation}
\Delta (\rho_{ABC}) \defeq \inf_{\sigma_{ABC}\in \mathcal{M}_{AB|C}} D_1 (\rho_{ABC} \|  \sigma_{ABC})
\end{equation}
poses an upper bound to the quantum conditional mutual information, i.e.
\begin{equation}
 I (A:B |C)_{\rho}\leq \Delta (\rho_{ABC})
\end{equation}
and there exist certain states for which the inequality becomes strict. Along these lines we will show that it is not possible to generalize the quantum conditional mutual information to the non-i.i.d.~setting by optimizing quantum R\'enyi divergences over the set of all Markov states. This will be done in the following by the construction of a specific example, where this approach leads to an upper bound of the quantum conditional mutual information for any value of $\alpha$. Note that this approach may seem reasonable, as it can be shown that
\begin{equation}
\Delta_{\min} (\rho_{ABC}) \defeq \inf_{\sigma_{ABC} \in \mathcal{M}_{AB|C}} - \textnormal{log} F^2( \rho_{ABC}, \sigma_{ABC})
\end{equation}
shares three main properties with the quantum conditional mutual information. Namely, these are non-negativity, non-increasing under data processing and duality for quadripartite pure states \cite{Erk14} (we review the argument in Appendix A). \\
Generalizing the above approach with the quantum R\'enyi divergences, we define
\begin{equation}
\Delta_{\alpha} (\rho_{ABC}) \defeq \inf_{\sigma_{ABC} \in \mathcal{M}_{AB|C}} D_{\alpha} (\rho_{ABC} \|  \sigma_{ABC}).
\end{equation}

 for $\alpha \geq 0 $. Our result now reads as follows

\begin{thm}
Let $P_k$ be the projector onto the antisymmetric subspace $ \Lambda^k (\mathbb{C}^d)$ in $(\mathbb{C}^d)^{\otimes k}$ with $\mathcal{H}_A \cong \mathcal{H}_B \cong \mathbb{C}^d$ being the first two tensor factors and let $\rho_{ABC}\defeq \frac{P_k}{d_k}$, where $d_k \defeq \dim  \Lambda^k (\mathbb{C}^d) = \binom{d}{k}$.
Then
\begin{equation}
 \Delta_{\alpha} (\rho_{ABC}) \geq \log \sqrt{\frac{4}{3}}.
\end{equation}
\end{thm}

\begin{proof}
We start out from the definition of the left-hand side making use of the monotonicity under CPTP maps of the quantum R\'enyi divergences by tracing out subsystem C.

\begin{align}
\Delta_{0} (\rho_{ABC}) &=  \inf_{\sigma_{ABC} \in \mathcal{M}_{AB|C}} \lim_{\alpha \rightarrow 0^{+}} D_{\alpha} (\rho_{ABC}\|  \sigma_{ABC}) \\
&\geq   \inf_{\sigma_{AB} \in \mathcal{S}_{A|B}}\lim_{\alpha \rightarrow 0^{+}} D_{\alpha} (\gamma_d \|  \sigma_{AB}) \\
&=  \inf_{\sigma_{AB} \in \mathcal{S}_{A|B}} -  \textnormal{log }  \textnormal{Tr }  P_{\gamma_d} \sigma_{AB}
\end{align}
Due to the special form of the Markov states (see eq. \ref{markov}) it holds that the optimization in the last two lines runs over the set of all separable states  $\mathcal{S}_{A|B}=\{ \sigma_{AB} \in \mathcal{D} (\mathcal{H}_{AB}): \sigma_{AB} = \sum_j p_j \sigma^{(j)}_A \otimes \sigma^{(j)}_B \}$ where $\{p_j\}$ is some probability distribution. Moreover it is clear that $\rho_{AB} =\textnormal{Tr}_C \rho_{ABC}$ equals the totally antisymmetric state $\gamma_d$ in $(\mathbb{C}^d)^{\otimes 2}$. \\
Since the totally antisymmetric state $\gamma_d$ is invariant under the action $g \otimes g$, where $g$ is unitary, we can restrict the optimization problem to states obeying the same symmetry. It has been shown that the expression in the last line has a constant lower bound equal to $\log \sqrt{\frac{4}{3}}$ (see proof of Corollary 3 in \cite{CSW}). Observing that the quantum R\'enyi divergences monotonically increase in $\alpha$ concludes the proof.\\
\end{proof}

We can now directly relate the quantum conditional mutual information to $\Delta_{0} $ for the state $ \rho_{ABC}$ defined above.

\begin{cor}
Let $\rho_{ABC} = \frac{P_k}{d_k}$  as above then

\begin{equation}
 \Delta_{0} (\rho_{ABC}) >  I (A:B |C)_{\rho_{ABC}},
\end{equation}
for $k= \lceil \frac{d + 1}{2} \rceil $ and $d \geq 27$.
\end{cor}
\begin{proof}
An explicit evaluation of the right-hand side (see Ref. \cite{CSW}) gives us
\begin{align}
I (A:B |C)_{\rho_{ABC}}= \begin{cases} 2 \log \frac{d+2}{d} \ &\mbox{if } d \textnormal{ is even }\\
\log \frac{d+3}{d-1} & \mbox{if } d\textnormal{ is odd}.  \end{cases}
\end{align}
Finally a simple numerical calculation shows that 
\begin{align}
\log \sqrt{\frac{4}{3}} > \begin{cases} 2 \log \frac{d+2}{d} \ &\mbox{if } d \textnormal{ is even }\\
\log \frac{d+3}{d-1} & \mbox{if } d\textnormal{ is odd}  \end{cases}
\end{align}
is the case for $d\geq 27$.
\end{proof}
This implies that the relative gap between $\Delta_0$ and the quantum conditional mutual information can be made arbitrarily big by scaling up the dimension as
\begin{equation}
 \Delta_{0} (\rho_{ABC}) \geq const.
\end{equation}
while
\begin{equation}
I (A:B |C)_{\rho_{ABC}} \leq \frac{4}{d-1}=\mathcal{O} \left( \frac{1}{d} \right)
\end{equation}
when $k= \lceil \frac{d + 1}{2} \rceil $. 

\begin{cor}
Let $\rho_{ABC} = \frac{P_k}{d_k}$  as above. Then it holds

\begin{equation}
 \Delta_{\min} (\rho_{ABC}) >  I (A:B |C)_{\rho_{ABC}},
\end{equation}
for $k= \lceil \frac{d + 1}{2} \rceil $ and $d \geq 27$.
\end{cor}
\begin{proof}
The statement follows directly from the observation that $D_{\min} (\rho \|  \sigma)  \geq D_{0} (\rho \|  \sigma) $ (Lemma 5.2 in \cite{milan}).\\
\end{proof}

As in the case of  R\'enyi generalizations of some information measures, the property of convergence to their von Neumann equivalents in the i.i.d.~-limit emerges from an additional limit taken in the smoothing parameter, we show here that our results above also hold in a regularized and smoothed version.\\
First, we need to define a metric on $\mathcal{D} (\mathcal{H})$, which we choose to be the  purified distance (introduced in Ref. \cite{Tom10}), i.e.
\begin{equation}
P (\rho , \sigma)\defeq \sqrt{1-F^2(\rho , \sigma)} .
\end{equation}
Furthermore, two states $\rho$ and $ \sigma$ will be called $\epsilon$-close if and only if $P (\rho , \sigma)\leq \epsilon$. Also we define the ball of $\epsilon$-close states around $\rho \in \mathcal{D} (\mathcal{H}) $ as
\begin{equation}
\mathcal{B}^{\epsilon} (\rho) \defeq \{ \rho' \in \mathcal{D}:  P (\rho , \rho')\leq \epsilon \} .
\end{equation} 
Now we can define the smoothed version of the $\Delta_{0}$-quantity

\begin{equation}
\Delta^{\cdot,\epsilon}_{0} (\rho) \defeq \max_{\tilde{\rho} \in \mathcal{B}^{\epsilon} (\rho) } \Delta_{0} (\tilde{\rho})
\end{equation}
One can immediately observe that 
\begin{equation}
\label{smooth}
\Delta^{\cdot,\epsilon}_{0} (\rho) \geq \Delta_{0} (\rho) \quad \forall \, 0 \leq \epsilon < 1
\end{equation}
holds. As for the case of regularization, we need to generalize our statements to the case where the number of available copies of the given state goes to infinity. Therefore we define the regularized version of $\Delta_{0}$, i.e.
\begin{equation}
\Delta^{\infty,\cdot}_{0} (\rho) \defeq \lim_{n \to \infty} \, \frac{1}{n} \, \Delta_{0} (\rho^{\otimes n}) = \lim_{n \to \infty} \, \frac{1}{n} \, \inf_{\sigma \in \mathcal{M}_{A^{n}B^{n}|C^{n}}} D_{0} (\rho^{\otimes n} \|  \sigma).
\end{equation}
Thus we need a statement more general than Theorem 1, which reads as follows;

\begin{thm}
\label{reg}
Let $\rho_{ABC} = \frac{P_k}{d_k}$  as above. Then
\begin{equation}
 \Delta_{0} (\rho_{ABC}^{\otimes n}) \geq n \log \sqrt{\frac{4}{3}}.
\end{equation}

\end{thm}
\begin{proof}
Tracing out system C it holds that

\begin{align}
\Delta_{0} (\rho_{ABC}^{\otimes n}) &=  \inf_{\sigma_{ABC} \in \mathcal{M}_{A^{n}B^{n}|C^{n}}} \lim_{\alpha \rightarrow 0^{+}} D_{\alpha} (\rho_{ABC}^{\otimes n}\|  \sigma_{ABC}) \\
&\geq   \inf_{\sigma_{AB} \in \mathcal{S}_{A^{n}|B^{n}}}\lim_{\alpha \rightarrow 0^{+}} D_{\alpha} (\gamma_d^{\otimes n} \|  \sigma_{AB}) \\
&=  \inf_{\sigma_{AB} \in \mathcal{S}_{A^{n}|B^{n}}} -  \textnormal{log }  \textnormal{Tr }  P^{\otimes n}_{\gamma_d} \sigma_{AB}
\end{align}
The observation that the last line has a lower bound given by $n \log \sqrt{\frac{4}{3}}$ (see Lemma 9 and 12  in  Ref. \cite{CSW}) concludes the proof.
\end{proof}

Putting together the observations Theorem \ref{reg} and  eq. \ref{smooth} we can conclude that the results achieved for $\Delta_{0}$ also hold for its smoothed and regularized version, hence establishing

\begin{equation}
 \Delta^{\infty,\epsilon}_{0} (\rho_{ABC}) \defeq \lim_{n \to \infty} \, \frac{1}{n} \, \Delta^{\epsilon}_{0} (\rho^{\otimes n}) \geq const.  \quad \forall \, 0 \leq \epsilon < 1.
\end{equation}

\section{Conclusions and Open Questions}

Two important conclusions can be drawn from our results. On the one hand we have shown that the approach of optimizing quantum R\'enyi divergences over the set of all Markov states cannot lead to a quantity that poses a lower bound to the quantum conditional mutual information. For a generalization of the quantum conditional mutual information to the non-i.i.d. setting it would be desirable to have an ordering of the kind
$$I_{\max} (A:B |C)_{\rho} \geq I(A:B |C)_{\rho} \geq I_{\min} (A:B |C)_{\rho}.$$ In Ref. \cite{BKM14} the authors present a promising approach for accomplishing this task. In their work they also state a conjecture for which a proof would be a major step in the direction of finally achieving a proper generalization of the quantum conditional mutual information.\\
On the other hand, our results strengthen the observation that a small quantum conditional mutual information does not imply that the state is near to a quantum Markov state.\\
It should also be mentioned that recently the quantum R\'enyi divergences have been generalized \cite{WWY14,Mul13} and also extended to a two parameter family \cite{Aud13}. So as long as it is not clear how to generalize the quantum conditional mutual information it remains an interesting open question if it can be shown that a generalization in the above mentioned way is not possible for the whole range of the new parameters. One might also pose the question how the generalizations proposed in Ref. \cite{BKM14} relate to the observations of this paper having in mind the results given in Ref. \cite{Lin14}.

\subsection*{Acknowledgements}
We would like to thank Stefan B\"auml, Normand Beaudry, Mario Berta, Matthias Christandl, Omar Fawzi, Marcus Huber, Renato Renner and Andreas Winter for helpful discussion. We thank Milán Mosonyi for helpful comments on the final versions of the manuscript. PE is grateful for support through Exzellenz-Auslandsstipendien 2013, LIQUID and 369.science. PE's work is supported by the Spanish MINECO  Project No. FIS2013-40627-P, the Generalitat de Catalunya CIRIT Project No. 2014 SGR 966, Swiss National Science Foundation (SNF), the National Centres of Competence in Research “Quantum Science and Technology” (QSIT), the COST action on Fundamental Problems in Quantum Physics and the European Commission Project RAQUEL.

\appendix

\section{}

In the following we prove that the $\Delta_{\min}$ is non-increasing under data processing and fulfills a duality relation for quadripartite pure states \cite{Erk14}.

\begin{theo}
Let $\rho_{ABCD}$ be a pure state and let $\mathcal{E}$ be a trace non-increasing CPM from $B \rightarrow B'$ then 
\begin{equation}
 \Delta_{\min} (A:B | C)_{\rho}= \Delta_{\min} (A:B | D)_{\rho}
\end{equation}
and
\begin{equation}
 \Delta_{\min} (A:B | C)_{\rho} \geq  \Delta_{\min} (A:B' | C)_{\rho}.
\end{equation}
\end{theo}

\begin{proof}
Starting out from the definition, we exploit Uhlmann's theorem \cite{Uhl} to conclude that
\begin{align}
\Delta_{\min}  (A:B | C) &= \inf_{\substack{ \sigma_{ABC} \in \mathcal{M} }} - \log \| \sqrt{\rho_{ABC}} \sqrt{\sigma_{ABC}} \|^{2}_{1} \\
&=\inf_{\substack{ \sigma_{ABCDE} \\ I(A:B|C)_{\sigma}=0}} - \log \| \sqrt{\rho_{ABCDE}} \sqrt{\sigma_{ABCDE}} \|^{2}_{1} 
\end{align}

 where $\sigma_{ABCDE}$ is a purification of $\sigma_{ABC}$. Note that E is an ancillary system in case D is not of sufficiently high dimension to purify $\sigma_{ABC}$ and that due to the pureness of $\rho_{ABCD}$ 
\begin{equation}
\rho_{ABCDE}= \rho_{ABCD} \otimes \phi_E ,
\end{equation}
where $\phi_E$ is some arbitrary but fixed pure state. Starting from the right hand side we get that
\begin{align}
\Delta_{\min} (A:B | D) &= \inf_{\substack{ \sigma_{ABD} \in \mathcal{M} } } - \log \| \sqrt{\rho_{ABD}} \sqrt{\sigma_{ABD}} \|^{2}_{1} \\
&=\inf_{\substack{ \sigma_{ABCDE} \\ I(A:B|D)_{\sigma}=0}} - \log \| \sqrt{\rho_{ABCDE}} \sqrt{\sigma_{ABCDE}} \|^{2}_{1}.
\end{align}
 As both optimizations run over $\sigma$ as a purification of a Markov state while $\rho_{ABCDE}$ is fixed, they consequently lead to the same optimal state $\sigma_{ABCDE}$, which concludes the proof for the first part of the theorem.\\
For the second part of the theorem note that due to the properties of the trace norm 
\begin{equation}
\| \sqrt{\rho_{ABC}} \sqrt{\sigma_{ABC}} \|_{1} \leq  \| \sqrt{\mathcal{E} ( \rho_{ABC})} \sqrt{\mathcal{E} (\sigma_{ABC}) }\|_{1}.
\end{equation}
The sought for result follows as $\mathcal{E}$ does not affect the constraints of the optimization.

\end{proof}

\end{document}